\documentclass[11pt,reqno]{amsart}
\usepackage{amssymb,mathrsfs}
\usepackage{amsaddr}
\usepackage{graphicx,epsfig}
\usepackage[hypertex]{hyperref}
\usepackage[numbers,sort&compress]{natbib}
\usepackage{enumitem}
\usepackage{MnSymbol}

\newcommand{\ket}[1]{\mathop{|#1\rangle}\nolimits}

\newcommand{\kbr}[2]{| #1\rangle\!\langle #2 |}

\newcommand{\scri}{{\mathscr I}}

\newcommand*{\at}{@}

\newcommand{\ra}{\rangle}
\newcommand{\la}{\langle}
\newcommand{\be}{\begin{equation}}
\newcommand{\ee}{\end{equation}}

\newtheorem{thm}{Theorem}

\newtheorem{lem}[thm]{Lemma}

\def\bbC{\mathbb{C}}

\def\bbR{\mathbb{R}}

\newcommand{\nn}{\nonumber}
\def\wh{\widehat}
\def\dg{\dagger}
\def\df{\overset{\rm df}{=}}
\newcommand{\Tr}[1]{\mathop{{\mathrm{Tr}}_{#1}}}
\newcommand{\Cl}[1]{\mathcal{C}l_{1,#1}}
\newcommand{\whCl}[1]{\wh{\mathcal{C}l}_{1,#1}}
\newcommand{\id}{\mathop{{\mathrm{id}}}\nolimits}

\def\a{\alpha}
\def\b{\beta}
\def\g{\gamma}

\def\e{\epsilon}

\def\vr{\varrho}
\def\vp{\varphi}

\def\om{\omega}
\def\Om{\Omega}
\def\s{\sigma}

\def\om{\omega}
\def\Om{\Omega}
\def\gom{r_\omega}
\def\gom{g_\omega}

\def\N{\mathcal{N}}
\def\E{\mathcal{E}}

\def\D{\mathcal{D}}

\def\M{\mathcal{M}}
\def\T{\mathcal{T}}

\begin{document}
\title{The capacity of black holes to transmit quantum information}

\begin{abstract}
    We study the properties of the quantum information transmission channel that emerges from the quantum dynamics of particles interacting with a black hole horizon.   We calculate the quantum channel capacity in two limiting cases where a single-letter capacity is known to exist: the limit of perfectly reflecting and perfectly absorbing black holes. We find that the perfectly reflecting black hole channel is closely related to the Unruh channel  and that its capacity is non-vanishing, allowing for the perfect reconstruction of quantum information outside of the black hole horizon. We also find that the complementary channel (transmitting entanglement behind the horizon) is entanglement-breaking in this case, with vanishing capacity. We then calculate the quantum capacity of the black hole channel in the limit of a perfectly absorbing black hole and find that this capacity vanishes, while the capacity of the complementary channel is non-vanishing instead. Rather than inviting a new crisis for quantum physics, this finding instead is in accordance with the quantum no-cloning theorem, because it guarantees that there are no space-like surfaces that contain both the sender's quantum state and the receiver's reconstructed quantum state.
\end{abstract}


\author{Kamil Br\'adler}
\email{kbradler\at ap.smu.ca}
\address{
    Department of Astronomy and Physics,
    Saint Mary's University,
    Halifax, Nova Scotia, B3H 3C3, Canada
    }
\author{Christoph Adami}
\address{Department of Physics and Astronomy, Michigan State University, East Lansing, MI 48824}

\maketitle

\section{Introduction}

The black hole information paradox has remained at the forefront of theoretical physics through close to 40 years, ever since Hawking discovered the eponymous radiation~\cite{hawking1975particle,Wald1994}. The information paradox comes in many shapes and forms, and what is considered paradoxical about black holes has changed significantly throughout history~\cite{Mathur2009,mathur2011information,gomberoff2005lectures,papantonopoulos2009physics}. In the past, even the idea that quantum black holes cause pure states to evolve into mixed states was considered paradoxical, but the advent of quantum information theory~\cite{nielsen2010quantum,wilde2011classical} has established that pure states very naturally evolve into mixed states if part of the Cauchy surface is traced over. More modern discussions of the black hole paradox discuss the fate of classical and quantum information when interacting with the black hole horizon. As the black hole is a quantum object, the use of the modern tools of quantum information theory are inevitable. However, none of these tools existed when the fate of information interacting with black holes was first discussed.

We propose that the fate of both classical and quantum information should be studied using the language of quantum Shannon theory,  the subdiscipline of quantum information theory  dealing with the mathematical aspects of information transmission, storage and retrieval~\cite{holevo2012quantum}. We argue that a resolution of the alleged paradox does not lie in our (incomplete) understanding of quantum gravity (as is often speculated), but rather in posing the problem correctly within an appropriate formalism dealing with information. Indeed, as the recent upswell of interest in the so-called firewall problem~\cite{almheiri2013black} (see also~\cite{braunstein2013better}) shows, the attempts to solve it are increasingly focused on the quantum information aspect of the problem, see for example~\cite{larjo2012black,giddings2012nonviolent,avery2012unitarity,nomura2013complementarity,papadodimas2013state}. While the firewall paradox is close to the subject matter we discuss here, will be addressed directly in a separate publication~\cite{adamibradler}.

Whether or not quantum black holes destroy {\em classical information} (that is, classical information carried by quantum states~\cite{holevo2012quantum}) has been answered negatively in terms of quantum information theory by Adami and ver Steeg~\cite{AdamiVersteeg2013}, who showed that the classical capacity of the black hole (the Holevo capacity~\cite{holevo2012quantum}) is positive. Indeed, these authors showed that the classical information sent into a black hole is not contained within Hawking radiation (as is often suspected), and therefore cannot be retrieved from it. Instead, the information resides in the {\em stimulated} emission of radiation that is unavoidable if quanta are absorbed by a black body, and can be reconstructed from this radiation with perfect accuracy if appropriate error protection measures are taken (this is necessary to recover information in any noisy channel). However, knowing the fate of classical information does not immediately shed light on what happens to {\em quantum} information interacting with a black hole, because quantum information has very different properties.

Quantum information refers to the entanglement state of a sender A (conventionally referred to as ``Alice"). In a quantum transmission channel, Alice intends to transmit her entanglement with a reference system to a receiver B (commonly called ``Bob") so that after the transmission Bob is entangled with the reference in exactly the same way as Alice was initially. Alice cannot retain this entanglement after transmission, because (as we will discuss in more detail) this would violate the celebrated quantum no-cloning theorem~\cite{wootters1982single,dieks1982communication}.

Hayden and Preskill  were first to cast the problem of {\em quantum information} interacting with black holes in terms of quantum Shannon theory~\cite{hayden2007black} (also see~\cite{smolin2006locking}), and discussed precisely the scenario we just outlined: how entanglement between Alice and a reference system could be transferred to Bob after it had interacted with a black hole. Contrary to their approach (that is similar in spirit, but differs in detail from the model of black hole evaporation advocated by Page~\cite{page1993information}) we describe the black hole as an open quantum system whose purifying (reference) quantum system is not accessible, and may as well be purely formal. As such, the black hole is described by a quantum channel, while its evolution is perfectly unitary.

Our first task will be to identify the nature of the channel, and then calculate the quantum capacity in those cases where this is possible today. Only then can we ask precise questions about the nature of the quantum channel: how much quantum information can be recovered from a black hole at future infinity, and how much quantum information enters the black hole.  Quantum Shannon theory gives a precise operational meaning to the notion of quantum information and quantifies it in terms of the amount of entanglement that can be transferred (from Alice to Bob) by the channel. This quantity is the {\em quantum channel capacity}.

Here we calculate the capacity of the black hole to transmit quantum information (in the form of shared quantum entanglement) {\em through} a black hole, in the semiclassical picture of quantum gravity.
We show that the capacity to reconstruct a quantum state depends on how well the black hole {\em reflects} it, and calculate the capacity explicitly in the limit of a perfectly absorbing black hole, as well as a perfectly reflecting black hole. We show that in neither case is it possible to reconstruct the quantum state accurately both outside and inside the horizon, in accordance with the no-cloning theorem of quantum mechanics. Our findings also imply that the loss of quantum information ought not to be viewed as a breakdown of quantum mechanics, and argue that the apparent evolution of pure states into mixed states is described by {\em open system dynamics} (namely, quantum channels) and is in accord with standard quantum mechanics.

The material we present is organized as follows.  In Sec.~\ref{sec:BHchannel} we review the definition of our main tool---the quantum channel capacity. We then describe the black hole channel (along with its relation to the Unruh channel) and review the notions of complementary and degradable channels. In Sec.~\ref{sec:Sorkin}, we introduce the black hole channel proper in terms of a construction due to Sorkin~\cite{sorkin1987simplified}, who discovered how to consistently extend the standard formalism of semi-classical quantum gravity to black holes with arbitrary reflectivity\footnote{Hawking introduced grey-body factors in his original derivation, however this formulation was not consistent with Einstein's standard results concerning the quantum theory of radiation, as discussed in~\cite{bekenstein1977einstein})}. This formulation allows us to discuss the fate of quantum information incident on an already formed black hole. We discuss the quantum channel in the limit of perfectly reflecting as well as perfectly absorbing black holes, followed by conclusions and an appendix containing a technical lemma needed in Sec.~\ref{sec:Sorkin}.

\section{Quantum channels and quantum capacity}
\label{sec:BHchannel}

Before we identify the precise nature of the quantum black hole channel and quantify how much quantum information can be recovered after interacting with a black hole, let us first illustrate the methods of quantum Shannon theory by showing why no information can be hidden in the outgoing Hawking radiation~\cite{hawking1975particle}.

That Hawking radiation is featureless (and because of this cannot carry the imprint of the quantum states that have interacted with the black hole) is on the one hand well known, yet disputed as a matter of principle.  We will confirm that Hawking radiation is featureless using the formalism of quantum information theory (quantum Shannon theory in particular) because this is the formalism we will be using in a channel construction that goes beyond Hawking's standard results. A reader less familiar with the basic concepts of quantum information theory is encouraged to consult Refs.~\cite{nielsen2010quantum,holevo2012quantum,wilde2011classical} for a rigorous introduction to quantum Shannon theory.
\begin{figure}[t]
  \resizebox{14cm}{!}{\includegraphics{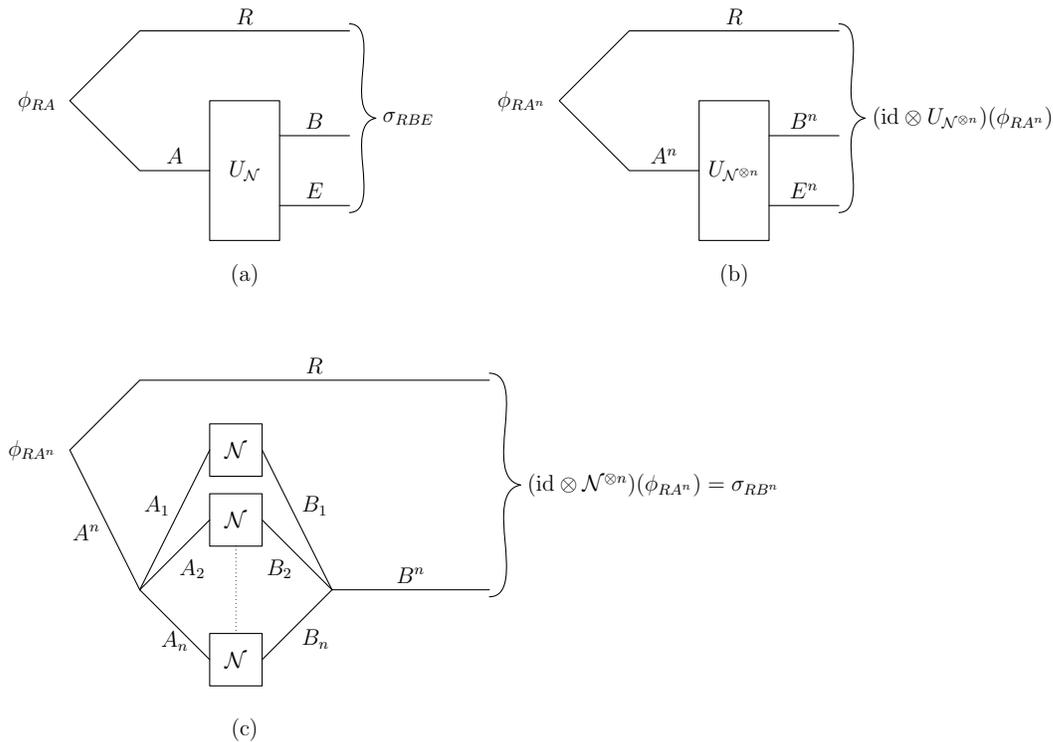}}
   \caption{(a): States and subsystem used for the calculation of the one-shot optimized coherent information Eq.~(\ref{eq:OptCohInfo}). An input pure state $\phi_{RA}$ is unitarily transformed by the isometry $U_\N$ representing a noisy quantum channel $\N:\phi_A\mapsto\s_B$. The $E$ subsystem is usually called a complementary output to $B$. (b): the same situation, but this time with $n$ copies of a quantum channel $\N$ where  the maximization in Eq.~(\ref{eq:OptCohInfo}) should be performed over the input state $\phi_{RA^n}$  to be used for the quantum capacity calculation Eq.~(\ref{eq:QuantCap}). This often becomes  intractable as $n$ increases. The state $\phi_{RA^n}$ is called a quantum codeword and is itself an output of another map called an encoder $\E$ as indicated in Eq.~(\ref{eq:map}). The output state $\Tr{RE^n}[(\mathrm{id}\otimes U_{\mathcal{N}^{\otimes n}})(\phi_{RA^n})]$ is then sent to the decoder $\D$. The symbol $\id$ stands for a noiseless (identity) channel and $S^n$ is an abbreviation for an $n$-partite system $S_1S_2\dots S_n$. (c): a more detailed version of the figure from (b), but with the complementary channel outputs suppressed.}
    \label{fig:cohinfo}
\end{figure}
Our discussion of black holes in this paper will be entirely within the semiclassical framework, meaning that we consider (as usual) macroscopic Schwarzschild black holes where the effects of backreaction (the influence of the black hole on the metric field surrounding it) as well as recoil (that is, momentum conservation) are neglected. As is furthermore customary, we assume that the black hole emits a scalar massless field only~\cite{BirrellDavies1982}, and neglect the gravitational redshift (while pointing out important places where including the redshift is expected to make a qualitative difference). While these simplifications result in a caricature of a physical black hole~\cite{page1976particle}, we do not expect that relaxing them, or considering other types of quantum fields, would change our conclusions in a qualitative manner (as long as  superselection rules for higher-spin fields are properly taken into account~\cite{bradler2011capacities}).

In order to determine how much information is encoded within Hawking radiation alone, we should calculate the {\em quantum capacity}~\cite{lloyd1997capacity,shor2002quantum,devetak2005capacity,hayden2008decoupling,klesse2008random}
of the {\em quantum channel}~$\N$. The term ``quantum channel" is synonym for a completely positive map that represents a general noisy evolution of positive operators (density matrices) in quantum mechanics~\cite{nielsen2010quantum}. In quantum information theory we ask how the sender (Alice) can transmit her quantum information reliably to the receiver (Bob) through a noisy quantum channel so that after Bob receives the output of the channel, his entanglement is precisely the same as Alice's was. This question is unambiguously answered using the concept of quantum capacity~\cite{holevo2012quantum}. If the quantum channel $\N$ is not too noisy (this statement can be made precise) then there exist two other quantum maps called an encoder $\E$ and decoder $\D$ such that the composite channel
\begin{equation}\label{eq:map}
\D\circ\N^{\otimes n}\circ\E
\end{equation}
is arbitrarily close to a noiseless  map (the identity operator) under a suitable norm called the diamond norm\footnote{\url{https://cs.uwaterloo.ca/~watrous/LectureNotes.html}}.

The existence of the encoder
and decoder is equivalent to the existence of an error-correcting code enabling the participants to
communicate in an error-free manner in the asymptotic limit of many copies of the channel~$\N$~\cite{shor1995scheme,gottesman1997stabilizer}. The quantum capacity is then the maximum of the ratio of the number of qubits we intend to send to the number of qubits the encoder
generates, to make sure that a given quantum system can be faithfully decodable (that is, correctable) by the recipient.
Thus, the quantum capacity is a number between zero and one, with units (qu)bits. It turns out
that the dimensionality of the Hilbert space available for the error-free transmission is approximately $2^{nQ(\N)}$, where the exponent is the quantum capacity given by
\begin{equation}\label{eq:QuantCap}
Q(\N)=\lim_{n\rightarrow\infty}{\frac{1}{n}}Q^{(1)}(\N^{\otimes n}).
\end{equation}
The quantity $Q^{(1)}(\N)$ is called the {\em optimized coherent information}~\cite{barnum1998information} defined as
\begin{equation} \label{eq:OptCohInfo}
Q^{(1)}(\N)\df\max_{\phi_{RA}}I_c(\N)=\max_{\phi_{RA}}\big[H(B)_{\s}-H(E)_{\s}\big].
\end{equation}
$H(S_{1})_{\s}$ is the von Neumann entropy of the $S_{1}$ subsystem of an $n$-partite quantum state $\s_{S_{1}S_{2}\dots S_n}$, we implicitly define $H(B)_\sigma$ and $H(E)_\sigma$ on the state $\s_{RBE}=(\id\otimes U_\N)(\phi_{RA})$ in Eq.~(\ref{eq:OptCohInfo}),
which is the output of an isometric extension $U_{\N}$ of the quantum channel $\N$. $B$ refers to the Hilbert space of the receiver (Bob in our case), while $E$ refers to the (unmeasured) environment and $R$ is the purifying system such that $\s_{RBE}$ is a pure state (see Fig.~\ref{fig:cohinfo}). The maximization is performed over all possible entangled states $\phi_{RA}$.

As is by now well known, the calculation of the quantum capacity for an arbitrary
quantum channel is an intractable problem~\cite{holevo2012quantum,wilde2011classical}. However,
some classes of quantum channels exist for which the quantum capacity can be calculated nonetheless, by showing that the regularization in Eq.~(\ref{eq:QuantCap}) is unnecessary (such channels are said to have ``single-letter" quantum capacity formulas).
One such class is the symmetric quantum channel~\cite{smith2008quantum}, whose quantum capacity is provably zero~\cite{smith2008quantum}.
The symmetric quantum channel is a very special case of a broader class of channels called {\em degradable} channels that, perhaps surprisingly, play an important role in black hole quantum information theory. We will return to them later in this section.

As it turns out, the quantum channel whose output is thermal Hawking radiation is symmetric. To see this, we inspect the channel isometry
\begin{equation}\label{eq:HawkIs}
V(r_\om)=\prod_{\om}\exp{\big[r_\om(a_k^\dg b_{-k}^\dg-a_kb_{-k})\big]},
\end{equation}
where $r_\om$ is related to the mode frequency $\omega$ and the surface gravity $\kappa=1/2M$ of the black hole (where $M$ is its mass) via $\exp{(-\pi\om/\kappa)}=\tanh{r_{\om}}$.
In (\ref{eq:HawkIs}), $a_k$ annihilates the receiver's vacuum ($B$) while $b_k$ annihilates vacuum states beyond the black hole horizon ($E$). Here, all the mode information is collected in the momentum index $k$ and a dispersion relation holds: $\om=\om(\pm k)$. There is no need to distinguish between continuum and discrete normalization states.

The product in (\ref{eq:HawkIs}) is not infinite (we assume low- and high-frequency cut-offs whose motivation is now well understood~\cite{unruh1976notes,jacobson1993black}). In fact, a finite-product form carries a number of important advantages. By acting on an initial vacuum state (Alice's) we find
\begin{equation}\label{eq:HawkIso}
V(r_\om)\ket{\rm vac}=\prod_{\om}{\frac{1}{\cosh{r_{\om}}}}\sum_{n=0}^{\infty}
\tanh^{n}{r_{\om}}\ket{n}_B\ket{n}_E=\prod_{\om}\s_{\om,BE}
\end{equation}
where we set $B\equiv k$ and $E\equiv -k$, to connect with the capacity formulas introduced above. Because Eq.~(\ref{eq:HawkIso}) is in product form, we can focus on the single two-mode state $\s_{\om,BE}$ only. With
$\vr^{\rm th}_{\om,B}={\rm Tr}_E(\s_{\om,BE})$ and similarly for $\vr^{\rm th}_{\om,E}$,
the isometry (\ref{eq:HawkIs}) induces a quantum channel whose output density matrix reads
\begin{equation}\label{eq:HawkDensityMatrix}
\prod_{\om}\vr^{\rm th}_{\om,B}\df\prod_{\om}\big(1-e^{-{\frac{2\pi\om}{\kappa}}}\big)\sum_{n=0}^{\infty}
e^{-{\frac{2\pi n\om}{\kappa}}}|n\ra_B\la n|
=\prod_{\om}{\frac{1}{\cosh ^{2}{r_{\om}}}}\sum_{n=0}^{\infty}\tanh^{2n}{r_{\om}}|n\ra_B\la n|
\end{equation}
where ``th" stands for thermal. It is easy to see that $\vr^{\rm th}_{\om,B}=\vr_{\om,E}^{\rm th}$. This is sufficient to conclude that the corresponding quantum channel is symmetric and that its quantum capacity therefore necessarily vanishes~\cite{smith2008quantum}. Furthermore, the overall channel output $\vr_B^{\rm th}=\prod_{\om}\vr^{\rm th}_{\om,B}$ is in a product form and therefore the resulting multi-mode quantum channel is again symmetric. So as advertised, no quantum information can ever be reconstructed from Hawking radiation: the capacity to transmit quantum information via Hawking radiation vanishes.

To some extent, the result of the above calculation is hardly surprising: after all, the input Fock space in this channel is one-dimensional, and such a Hilbert space cannot be used to transmit information. Indeed, this is completely analogous to classical information theory: we need at least two degrees of freedom (for example, two voltage levels) to be able to transmit a physical bit.

What would happen if we attempted to increase the dimension of Fock space by viewing $\ket{\rm vac}$ as a logical zero $\ket{0}$ and consider the single photon $\ket{1}_k$ as the logical one $\ket{1}$? Doing this transforms the isometry into a different one-- thus inducing a completely different quantum channel. This is a reminder that particles alone (whether bosons or fermions) cannot by themselves be considered qubits because a particle alone does not specify the basis states that are necessary to construct the Bloch sphere representation of the qubit. Indeed, while most qubits are based on bosonic or fermionic states, the correspondence between particle states and qubits must be justified in terms of a low-level encoding that maps the logical $|0\ra$ and $|1\ra$ to physical particle states. Below we will revisit the isometry Eq.~(\ref{eq:HawkIs}) under less trivial circumstances, and a methodical approach rooted in quantum Shannon theory will help us to sharpen the information loss paradox and guide us toward its resolution.

For now we are led to the following conclusions: If quantum information is ever to escape the evaporating
black hole, the outgoing radiation needs to either
\begin{itemize}
 \itemsep0.5em
   \item[$\pentagram$] display some form of inter-mode entanglement within the output multi-mode entangled state, or
   \item[$\pentagram$] exhibit {\em non-thermal} corrections in the two-mode output state.
 \end{itemize}
Both scenarios are necessary but not sufficient to obtain a nonzero quantum capacity.

In this work, we will find that the second scenario is realized, while inter-mode entanglement (quantum correlations across modes with different $|k|$) does not appear natural given the physics involved. Indeed, Eq.~(\ref{eq:HawkIso}) guarantees that modes do not interact. Nonetheless, such approaches trying to reconcile information preservation with black hole dynamics have been tried before~\cite{giddings2012nonviolent,braunstein2007quantum}.

\subsection*{Black holes as quantum transmission channels}

The quantum channel in the presence of black holes was previously studied by Hayden and Preskill in~\cite{hayden2007black}. Let us consider two observers, Alice (the sender of quantum information) and Bob (the receiver). Alice sends her quantum state into the Schwarzschild black hole and Bob collects the radiation at future null infinity. His goal is to reconstitute the quantum states sent by Alice with arbitrary accuracy as measured by an operationally justified figure of merit (the quantum fidelity~\cite{nielsen2010quantum}). If this is possible, then quantum information processing by a black hole is manifestly unitary. If this is not possible in principle, then quantum black holes hide quantum information as long as the black hole horizon is present, but it does not imply a violation of unitarity. Whether black holes {\em destroy} quantum information can only be ascertained if we were able to describe the entangled system after the black hole has fully evaporated. Hayden and Preskill realized that the process of thermalization (called scrambling in the high-energy jargon) resembles a random code construction~\cite{hayden2008decoupling,klesse2008random,devetak2005capa}, which would imply automatic protection of quantum information thrown into black hole. We follow a similar path but contrary to~\cite{hayden2007black} we do not assume that this procedure could preserve quantum information under all circumstances. The main reason is that the same randomization operation producing a thermal state when applied to a vacuum {\em fails to generate a thermal state for an incoming $n$-photon state}\footnote{For a finite-dimensional Hilbert space $\bbC^d$, the encoding unitary operator modelling the thermalization is chosen randomly according to the Haar measure for the unitary group $U(d)\otimes U(d)$~\cite{hayden2008decoupling}. In an infinite-dimensional Hilbert space of a real scalar massless field, two options are possible: (i) to cut-off the Hilbert space and make it $\bbC^d$ for $d\gg0$ or (ii) to generalize the randomization procedure to an infinite-dimensional Hilbert space. In certain special instances of the latter case~\cite{harrington2001achievable,holevo2012quantum} it is indeed possible that if the incoming states are Gaussian, then Gaussian random codes (an ensemble of coherent states for example) are synonymous with thermalization and can be used to achieve the quantum capacity given by the coherent information~Eq.~(\ref{eq:OptCohInfo}). However, it is not known whether the Gaussian random codes are able to achieve nonzero rates for incoming Fock states (or other non-thermal quantum states) used to send quantum information in our scenario. The black hole is certainly not a thermalization machine because as we will see in Sec.~\ref{sec:Sorkin}, the scrambled state is simply not Gaussian (and so it is non-thermal).}. Thus, it is not immediately obvious that randomization is a good encoder (it is certainly true in the finite-dimensional qubit model studied in~\cite{hayden2007black} but not necessarily in a more realistic, infinite-dimensional Hilbert space we consider here). Instead, we consider this operation to be a quantum channel, for a detailed discussion see Eq.~(\ref{eq:map}).

Clearly, transfer of entanglement is not the experiment proposed by Hawking to illustrate the failure of quantum mechanics~\cite{hawking1975particle} (or even a reformulated setup to illustrate the loss of information as mentioned in the introduction). To some extent this is inevitable because the concept of quantum information did not exist at the time.
A simple thought experiment describes a basic quantum-information theoretic scenario that gives rise to an apparent entropy production. Suppose we prepare a set of pure states and then apply to it a random unitary operation, with a probability chosen from a given probability distribution. If the receiver of the quantum state is unaware of which unitary has been applied, she cannot reconstruct the initial quantum state with perfect accuracy. This is a perfectly unitary operation of course, and the resulting mixed state can be (non-uniquely~\cite{nielsen2010quantum}) purified using an appropriate ancilla (reference or auxiliary state). In fact, Bekenstein~\cite{bekenstein1973black} speculated that the source of the black hole entropy is an ensemble average over the many ways a black hole can be prepared. This is, mathematically speaking, equivalent to a series of random unitary operators applied to an initial quantum state described above. This entropy-producing series can be seen as analogous to the black hole evaporation process and (as emphasized by Zurek in~\cite{zurek1982entropy}) this mapping must be described by a superoperator and not a unitary operator. This is correct, of course,  since Zurek's notion of a superoperator is mathematically equivalent to what has later been termed a quantum channel (a completely positive map)~\cite{nielsen2010quantum}.

While the description of black hole evaporation as a series of random unitary operators applied to an unknown initial quantum state is clearly simplistic, we described this gedankenexperiment in order to place the process of black hole formation and evaporation firmly into the realm of quantum Shannon theory. We learn that unitary processes can give rise to channels that are described by superoperators, that entropy production is natural in such channels, and that an apparently mixed state (after evaporation of the black hole) should not be used as evidence of the failure of quantum mechanics.

\subsection*{The Unruh channel}

In the absence of an exact description of the black hole interior we assume here that
the principal physical mechanism accounting for quantum information preservation
is the process of black hole stimulated emission~\cite{bekenstein1993fast,wald1976stimulated,panangaden1977probability,AudretschMueller1992,sorkin1987simplified} (see Fig.~\ref{fig:stimul}). A similar approach has been adopted in~\cite{AdamiVersteeg2013} in the
context of the transmission of classical information via black holes. As is now known (and we review below), the dynamics of stimulated emission from a perfectly reflecting accelerated mirror (in the absence of backscattering) gives rise precisely to the channel known as the Unruh channel previously studied in great detail in~\cite{bradler2009private,bradler2010conjugate,bradler2011infinite,bradler2012rindler} in the context of the Unruh effect~\cite{unruh1976notes} (for a comprehensive review see~Ref.~\cite{takagi1986vacuum} and~\cite{crispino2008unruh}). As we will see below, the same channel arises when considering perfectly reflecting black holes~\cite{AdamiVersteeg2013}.
\begin{figure}
  \resizebox{12cm}{!}{\includegraphics{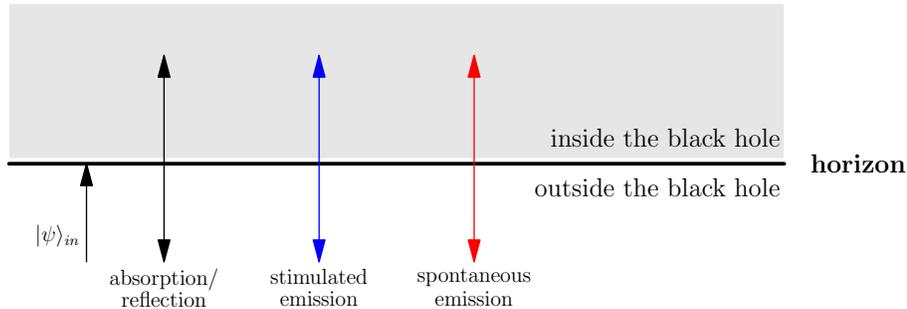}}
   \caption{A schematic description of the black hole's response to a {\em late} incoming state $\ket{\psi}_{in}$ of $n$ photons. The Hawking radiation is responsible for spontaneous emission (red arrows) that is ``modulated" by a potential barrier surrounding the black hole leading to a nonzero reflection or absorption coefficients (black arrows). This is a situation of a black hole interacting with a vacuum ($n=0$). In  case $n>0$, the black hole additionally emits stimulated radiation (blue arrows).}
    \label{fig:stimul}
\end{figure}

The quantum (and classical) information transmission properties of the Unruh channel are fully understood since the respective capacities are calculable in terms of single-letter capacities. Even more interestingly, the structure of the Unruh channel is intimately related to optimal qubit quantum cloners~\cite{bradler2011infinite}. We will return to this peculiarity after we review the properties of the Unruh channel. That the Unruh channel makes an appearance in the  discussion of black holes should not come as a surprise considering the close resemblance between the physical processes of the Unruh effect and black hole evaporation. However, there are also a number of differences and it would be hasty to identify the Unruh channel with the quantum black hole channel in general. Indeed, the full quantum channel is different from the limiting cases we discuss here, and was studied in detail elsewhere~\cite{bradler2014black}. Yet, we will see that in the limit where the black hole is perfectly reflecting incoming radiation, the black hole channel exactly coincides with the Unruh channel.

The Unruh effect~\cite{unruh1976notes} precedes the discovery of Hawking radiation. Originally described by Fulling~\cite{Fulling1973} and independently by Davies~\cite{Davies1975}, the effect Unruh considered concerns the radiation that an accelerated observer perceives when an observer at rest measures the vacuum: the absence of any particles. From a quantum information-theoretic point of view, the Unruh effect can be viewed as a quantum information transmission channel, where the sender is at rest and the receiver is accelerated with respect to the sender.

The formal input of the channel is a quantum state prepared in the laboratory of an inertial (Minkowski) observer. The output of the channel is the quantum state detected by a uniformly accelerating observer whose natural reference frame is described by Rindler coordinates $\xi,\tau$ ($a$ is the Rindler observer's proper acceleration):
\begin{subequations}\label{eq:RindlerCoord}
\begin{align}
  x &= \xi\cosh{\tau a}, \\
  t &= \xi\sinh{\tau a},
\end{align}
\end{subequations}
with the Rindler metric ${\rm d}s^2=-a^2\xi^2{\rm d}\tau^2+{\rm d}\xi^2$. In order to quantize a field using the degrees of freedom available to the non-inertial observer, we use Rindler coordinates Eq.~(\ref{eq:RindlerCoord}) covering separately the right ($\xi>0$) and left wedge ($\xi<0$) even though the Rindler observer is bound to just a single wedge.

The Unruh effect is a consequence of the inequivalent quantization of the field (here, massless scalar bosons) in the respective reference frames. The creation and annihilation field
operators in the inertial and accelerated frames are related by the Bogoliubov transformation:
\begin{equation}\label{eq:SymplecticTrans}
\begin{pmatrix}
b_{\Om}^{R} \\
{b_{\Om}^{L}}^{\dg}
\end{pmatrix}
=
\begin{pmatrix}
\cosh{r_{\Om}} & \sinh{r_{\Om}} \\
\sinh{r_{\Om}} & \cosh{r_{\Om}}
\end{pmatrix}
\begin{pmatrix}
d_{-\Om} \\
{d^{\dg}}_{\Om}
\end{pmatrix},
\end{equation}
where $\Om=\om/a$ is the (rescaled) Minkowski frequency $\om$ labeling the different modes, and
$\tanh{r_\Om}=\exp{(-\pi\Om)}$. The boson operators $d_{\Om}$ annihilate the Minkowski vacuum and $\{d_{\Om},d^{\dg}_{\Om}\}$ satisfy the canonical commutation relations. In Rindler spacetime, there are two sets of boson operators,  for the right ($R$) and left ($L$) wedges respectively: $\{b_{\Om}^{R},{b_{\Om}^{R}}^{\dg}\}$ and
$\{b_{\Om}^{L},{b_{\Om}^{L}}^{\dg}\}$ that separately satisfy the commutation relations. In contrast to the Minkowski operator set, they annihilate the Rindler vacuum. Because the field operators define the vacuum (and as a consequence the notion of particle), the Minkowski and Rindler observer
cannot agree on the particle content of their respective vacuum. The field operators $\{d_{\Om},d^{\dg}_{\Om}\}$ are usually called (perhaps confusingly) the Unruh modes, in spite of defining the particle content in Minkowski, not in Rindler, spacetime. Unruh modes have convenient algebraic properties~\cite{unruh1976notes}: for example,  each Rindler mode is related to only two Unruh modes labeled by $\pm\Om$ and the transformation~(\ref{eq:SymplecticTrans}) belongs to the real symplectic group $Sp(2,\bbR)$.

As previously mentioned, the full set of Rindler modes is supported on both space-time wedges but the
accelerating observer can access only the set of modes in his ``own" wedge, because he is causally disconnected from the opposing wedge. As a consequence, the dynamics in each of the wedges is quantum-mechanically incoherent: the spacetime that is inaccessible to the Rindler observer has to be traced over so that he effectively obtains a mixed state. Of course, this does not imply that unitarity is lost in the process. To see why, it is helpful to invoke the machinery of quantum information
theory applied to the Unruh channel. We can construct the channel by switching from the Heisenberg to the
Schr\"odinger picture. There, Eq.~(\ref{eq:SymplecticTrans}) takes the following form~\cite{takagi1986vacuum}:
\begin{equation}  \label{eq:MinkRinassignment}
\ket{n}\mapsto \prod_{\Om}{\frac{1}{\cosh^{1+n}{r_\Om}}} \sum_{m=0}^{\infty}\binom{n+m}{n}^{1/2}\tanh^{m}{r_{\Om}}\ket{n+m}_{\Om,L}\ket{m}_{\Om,R}.
\end{equation}
Eq.~(\ref{eq:MinkRinassignment}) relates the particle content using Minkowski modes (an $n$-particle state, left-hand-side) to the particle content in Rindler spacetime (right-hand-side). Because the overall multi-mode state Eq.~(\ref{eq:MinkRinassignment}) is a product state, we can without loss of generality focus on a single mode labeled by $\Om$, while assuming high and low frequency cut-offs to be present (see~\cite{takagi1986vacuum}).
Strictly speaking, each term labeled $\Om$ in Eq.~(\ref{eq:MinkRinassignment}) is referred to as a ``two-mode" state in quantum optics,  but as the factorization occurs for different modes labeled by $\Om$ rather than $(\Om,L)$ and $(\Om,R)$, we will refer to each term here as a the single mode $\Om$.

As a consequence, we can (in analogy with Eq.~(\ref{eq:HawkIso})) rewrite the transformation Eq.~(\ref{eq:MinkRinassignment}) in terms of the action of a single mode isometry $V_\Om^{A\to BE}$~\cite{bradler2012rindler}
\begin{equation}\label{eq:UnruhUnitaryeff}
V_\Om^{A\to BE}\big(\ket{n}_{A}\big)
={\frac{1}{\cosh^{1+n}{r_{\Om}}}}\sum_{m=0}^{\infty}\binom{n+m}{n}^{1/2}\tanh^{m}{r_\Om}\ket{n+m}_{B}\ket{m}_{E},
\end{equation}
where $A$ is the sender system, $B$ (the receiver system) is identified with the left Rindler wedge, and the environment (or reference, $E$) lies beyond the Rindler observer's horizon. The Unruh channel is defined by the isometry $V_\Om^{A\rightarrow BE}$, giving rise to a quantum channel
by tracing over the reference system $E$.

Let us now assume that Alice has at her disposal two modes labelled
$A_{1}$ and $A_{2}$ so she can prepare an arbitrary qubit in a Hilbert
space $\bbC^{2}$ spanned by $\{\ket{01}_{A_{1}A_{2}},\ket{10}_{A_{1}A_{2}}\}\equiv\{\ket{01}_{A},\ket{10}_{A}\}$
(for a qudit generalization see~\cite{bradler2012rindler}). In that case the isometry becomes
\begin{equation}\label{eq:UnruhDoubleIso}
V_\Om^{A_{1}\to B_{1}E_{1}}\otimes V_\Om^{A_{2}\to B_{2}E_{2}}=V_\Om^{A_{12}\to B_{12}E_{12}}.
\end{equation}
so that the qubit Unruh channel is defined as $\N\df\Tr{E_{12}}\big[V_\Om^{A_{12}\to B_{12}E_{12}}\big]$.
As shown in~\cite{bradler2011infinite}, the channel output can be written as
\begin{equation}
\N=\bigoplus_{\ell=1}^{\infty}p_{\ell}\,\N_{\ell},
\label{eq:UnruhOutput}
\end{equation}
where
$$
p_{\ell}=\frac12(1-z)^{3}\ell(\ell+1)z^{\ell-1}
$$
and $z=\tanh^{2}{r_{\Om}}\equiv\exp{(-{2\pi\om/a})}$  so  that $\sum_{\ell=1}^{\infty}p_{\ell}=1$. Eq.~(\ref{eq:UnruhOutput}) is usually termed a ``direct sum channel", ``probabilistic mixture of channels" or ``orthogonal convex sum channel"~\cite{holevo2012quantum}.

Let us investigate the structure of the channel's output Hilbert space $B_{12}$. For the purpose of encoding quantum information we will find it advantageous to switch from this bosonic Fock space to a bipartite, infinite-dimensional abstract Hilbert space isomorphic to a Hilbert space $\ell_2$ that we define below.
The Hilbert space $B_{12}$ is a bipartite space spanned by $\{\ket{m}_{B_1}\ket{n}_{B_2}\}_{m,n=0}^\infty$. A closer examination of the channel output (states from the Hilbert space $B_{12}$) reveals~\cite{bradler2011infinite,bradler2012rindler} that the channel output is actually confined to a completely symmetric subspace of $B_{12}$. This subspace has a direct sum structure $B_{12}=\bigoplus_{\ell=1}^\infty B_{12}^{(\ell)}$ where $\dim{B_{12}^{(\ell)}}=\ell+1$. Since the subspaces $B_{12}^{(\ell)}$ are completely symmetric, their spanning basis vectors are the set $\{\ket{m}_{B_1}\ket{\ell-m}_{B_2}\}_{m=0}^\ell$. The Hilbert space $B_{12}^{(\ell)}$ is bipartite but because the information is dual-rail encoded (see the qubit input Hilbert space and its spanning basis above Eq.~(\ref{eq:UnruhDoubleIso})), we must ignore its bipartite structure and we write the basis states simply as $\{\ket{m,\ell-m}_B\}_{m=0}^\ell$.  It is this completely symmetric subspace of $B_{12}$ that is isomorphic to the Hilbert space we will use to encode quantum information. This space is spanned by the standard basis  $e_n=(\dots,0,1,0,\dots)$ with zeros everywhere except for a single 1 in the $n$-th place. We will denote such a ``logical" basis state\footnote{The ``logical" basis is that which is used to encode quantum logic. For example, the logical zero could be encoded as $|0\ra=|01\ra_A$ and the logical one as $|1\ra=|10\ra_A$.}
by $\ket{n-1}$. With this identification, the output of the black hole channel $\N$ is just as in the Fock-space formalism, but the interpretation of the ket vectors is different: the set $\{\ket{n}\}_{n=0}^\infty$ simply denotes the basis of $\ell_2$.

\subsection*{Cloning channels}

Remarkably, the output states of $\N_{\ell}$  in Eq.~(\ref{eq:UnruhOutput}) give rise to prominent channels in quantum information theory called $1\to\ell$ cloning channels $\Cl{\ell}$ (see also~\cite{AdamiVersteeg2006}) because they yield $\ell$ approximate (identical) clones of an unknown input qubit (for $\ell=1$ the map is just an identity)~\cite{bradler2011infinite}. Cloning channels provide the best solution to the problem of cloning an unknown qubit, to a level allowed by the laws of quantum mechanics~\cite{wootters1982single,buvzek1996quantum,gisin1997optimal}. The quality of the clones is measured by a suitable figure of merit (the fidelity between an input state and one of the clones).  This family of channels is ubiquitous in quantum physics~\cite{mandel1983photon,scarani2005quantum} and  also played an important role recently in the proof of the generalized Wehrl conjecture~\cite{lieb2012proof}. In the light of Hayden and Preskill's speculation about a trade-off between the capacity of a black hole to clone or to destroy quantum information~\cite{hayden2007black}, it seems particularly opportune to study the properties of such channels. Note that quantum cloning channels do not  clone quantum states perfectly: the clones they produce are mixed states that are approximations of the original quantum state, with a fidelity $F$ that can be as high as $F=5/6$ for the optimal $1\to 2$ cloner~\cite{gisin1997optimal}. The literature also distinguishes clones from anti-clones~\cite{Buzeketal1999,CerfFiurasek2006}. The latter are the complex conjugate of the clones, and information-theoretically simply represent the best possible approximation of the orthogonal complement of a given pure state~\cite{Buzeketal1999}. In general, a cloning machine that attempts to create $\ell$ copies from $n$ inputs creates $n$ (approximate) clones
and $\ell-n$ anti-clones. \label{cloningdisc}

In the following we discuss the structure of these optimal cloning channels. One way to represent a quantum channel is by calculating the output density matrix of the channel~\cite{nielsen2010quantum}. We start with
\begin{equation}\label{eq:cloningoutput}
  \Cl{\ell}(\phi)=\frac2{\ell(\ell+1)}\bigg({\ell\over2}\id_{\ell+1}+\sum_{i=x,y,z}n_iJ^{(\ell+1)}_i\bigg),
\end{equation}
where $J^{(\ell+1)}_i$ is the $(\ell+1)$-dimensional representation of the generators of the $su(2)$ algebra (satisfying $\big[J^{(\ell+1)}_i,J^{(\ell+1)}_j\big]_-=i\e_{ijk}J^{(\ell+1)}_k$) and $\id_{\ell+1}$ is an identity matrix of dimension $\ell+1$.
For a single qubit in a dual-rail encoding
\begin{equation}\label{eq:dualrailInput}
\ket{\phi}_A=a\ket{01}_{A}+b\ket{10}_A
\end{equation}
and using  $\Cl{1}(\phi)=\phi$ for the $\ell=1$ channel,
we find $n_x=\bar ab+a\bar b,n_y=i(-\bar ab+a\bar b),n_z=|a|^2-|b|^2$, where the overbar denotes complex conjugation. It is remarkable that the coefficients $n_i$ do not depend on the representation $\ell+1$ of the generators of $su(2)$, that is, all cloning channels are described by the same coefficients $n_i$ irrespective of the number of clones produced. For all $\ell>1$ we just use the appropriate generators in~(\ref{eq:cloningoutput}) and obtain the corresponding output state forming the black hole channel as a convex combination of the cloners $\Cl{\ell}$. Before studying the structure of these channels in more detail, we need to discuss the concept of ``degradability".

\subsection*{Degradable channels}
Cloning channels are an important example of a type of quantum channel termed {\em degradable}.
The concept of degradability was introduced in~\cite{devetak2005capa} and studied in~\cite{cubitt2008structures}. To understand degradability we have to first define the {\em complementary channel} to a quantum
channel $\N$. Recall that in order to define a quantum channel $\N$ we start with the channel isometric extension $V_\Om^{A\to BE}$ and trace over the reference system $E$. The complementary channel $\wh\N$ to $\N$ is obtained  by tracing over the output system $B$ instead of the
reference~\cite{holevo2012quantum}: $\wh\N\df\Tr{B}\big[V_\Om^{A\to BE}\big]$. This is sufficient for the definition of degradability. A channel $\N$ is degradable if there exists another quantum channel $\D$ such that $\wh\N=\D\circ\N$. Then, $\D$ is called a {\em degrading} map.

If a channel is degradable, its quantum capacity (Eq.~(\ref{eq:QuantCap})) reduces to Eq.~(\ref{eq:OptCohInfo}), that is, calculating the capacity becomes a computationally tractable problem. It is not immediately obvious why for degradable channels the quantum (and other) capacities are calculable: they simply satisfy a certain entropic inequality which makes it possible to prove the inequality for the optimized coherent information:
\be
Q^{(1)}(\N\otimes\N)\leq 2Q^{(1)}(\N)\;.
\ee
Because at the same time $Q^{(1)}(\N\otimes\N)\geq 2Q^{(1)}(\N)$ follows from the definition (recall the optimization in Eq.~(\ref{eq:OptCohInfo})) we can conclude that degradable quantum channels have {\em additive} coherent information and therefore  $Q(\N)=Q^{(1)}(\N)$ from Eq.~(\ref{eq:QuantCap}). It is now known that qubit cloning channels are degradable~\cite{bradler2010conjugate} and that therefore their quantum capacity can be calculated: $Q(\Cl{\ell})=\log_{2}{\frac{\ell+1}{\ell}}$. Once we know the quantum capacity of $\Cl{\ell}$ for all $\ell$ we
can easily find the capacity of the Unruh channel~Eq.~(\ref{eq:UnruhOutput}). As shown in Lemma~\ref{lem:directsum} (see Appendix), the quantum capacity of a probabilistic
mixture of quantum channels is a probabilistic mixture of quantum capacities
(cf. Eq.~(\ref{eq:QdirectSumChannel}))
\begin{equation}
Q(\N)={\frac{(1-z)^{3}}{2}}\sum_{\ell=1}^{\infty}\ell(\ell+1)z^{\ell-1}\log_{2}
{\frac{\ell+1}{\ell}}.  \label{eq:UnruhQCap}
\end{equation}
where $z=\tanh^2{r_\Om}$. This result coincides with the derivation from~\cite{bradler2010conjugate}.

Curiously, it was not degradability but rather {\em conjugate degradability} that led to the result (\ref{eq:UnruhQCap}). Conjugate degradable channels are channels degradable up to complex conjugation (transposition of a density matrix in a given basis) and additivity of the optimized coherent information holds for conjugate degradable channels as well~\cite{bradler2010conjugate}. Conjugation is a nontrivial modification of the degradability criterion because transposition is not a quantum channel (it is a positive, but not completely positive map~\cite{nielsen2010quantum}). Cloning channels are both degradable and conjugate-degradable but the relationship between the two families of channels remains an open problem. For example, it is not clear whether there exist conjugate-degradable channels that are not degradable.

\section{Late-time quantum information interacting with a black hole}
\label{sec:Sorkin}
In this section we assume Alice, the sender of quantum information, to be poised just outside an already formed black hole and imagine her tossing the quantum message into the waiting abyss. Is there a model describing this kind of late-time interaction between quantum information and an already formed black hole? Indeed, this situation has been treated before~\cite{sorkin1987simplified} (see~\cite{AdamiVersteeg2013} for further analysis). Sorkin describes the interaction between a black hole formed by gravitational collapse and a late-time massless scalar quantum field. The late-time field's support is distinct from that of the Hawking radiation, which is red-shifted with respect to it.
As a consequence, we can study the interaction of late-time radiation with the black hole without mixing late-time with Hawking modes (their respective creation and annihilation operators commute).
\begin{figure}
  \resizebox{6cm}{!}{\includegraphics{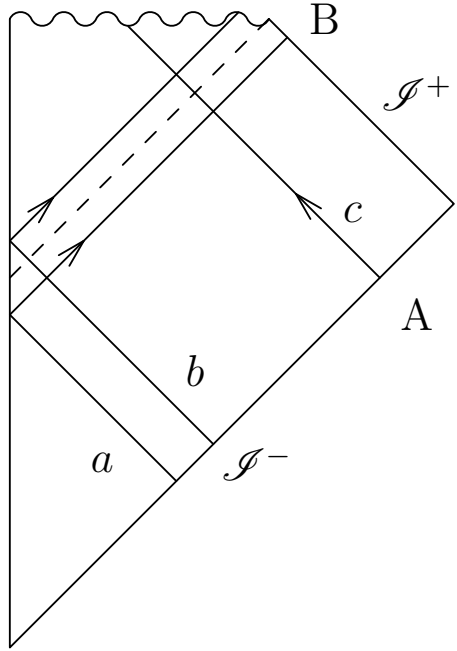}}
   \caption{Penrose diagram of an evaporating black hole formed by a gravitational collapse (ignoring back-reactions). The dashed line demarcates the horizon and the wavy boundary is the singularity. The interaction of the late-time incoming mode with the horizon is governed by Eq.~(\ref{eq:SorkinHamiltonian}). Alice (A) uses the late-time mode $c$ to send her quantum message inside the black hole. Bob (B) collects the radiation at $\scri^+$.}
    \label{fig:confdiag}
\end{figure}
The late-time isometry is derived from the following Bogoliubov transformation relating the input and output Hilbert spaces:
\begin{equation}\label{eq:SorkinHamiltonian}
  a_{\rm out}=\a a-\b b^\dg+\g c.
\end{equation}
Here, $a$ and $b$ annihilate early-time particles just outside and just inside the horizon as before, and $c$ annihilates late-time modes outside the horizon\footnote{Note that our notation differs from that of Sorkin~\cite{sorkin1987simplified}}. The coefficients  $\a,\b,\g\in\bbR$.

This mapping correctly predicts the effect of superradiance (an older term used for stimulated emission) even for a non-rotating black hole, reproducing earlier results~\cite{bekenstein1977einstein,panangaden1977probability}. Note that the superradiance discussed in~\cite{panangaden1977probability} is also relevant for non-rotating black holes~\cite{pananprivate}.
Using Sorkin's insight, the black hole can now be understood as a quantum system perturbed by incoming late-time radiation, whose response is outgoing Hawking-like radiation (schematically depicted in Fig.~\ref{fig:confdiag}). While we focus here on two special cases in which $\a=0$ (perfectly reflecting black hole) and $\a=1$ (perfectly absorbing black hole), a general analysis of the interaction described by Eq.~(\ref{eq:SorkinHamiltonian}) was presented elsewhere~\cite{bradler2014black}.

\subsection*{Perfectly reflecting black hole}

For $\a=0$ in Eq.~(\ref{eq:SorkinHamiltonian}), it turns out that the resulting isometry relating incoming and outgoing quantum states is just Eq.~(\ref{eq:HawkIs}),  where the late-time particles play the role of the early-time particles of the standard description, so that
$\tanh^2{r_\om}=\frac{\beta^2}{1+\beta^2}$ precisely like the standard black hole channel, but where now $\gamma^2=1+\beta^2$.

Clearly, this mapping is formally isomorphic to the Unruh isometry $V_\Om^{A\to BE}$ given by~Eq.~(\ref{eq:UnruhUnitaryeff}). Indeed, setting  $m=0$ in Eq.~(\ref{eq:UnruhUnitaryeff}) we can identify the left and right Rindler vacuum with the Boulware vacuum. We can further map $r_\Om$ (which depends on the proper acceleration $a$ of the receiver) to $r_\om$ (defined by the black hole's surface gravity $\kappa$) and the equivalence is exact\footnote{The isometry Eq.~(\ref{eq:UnruhUnitaryeff}) first appeared in the black hole context in~\cite{wald1976stimulated}. The mapping describes the stimulated emission of radiation in response to the events that formed the black hole. However, in order to be able to observe the stimulated radiation at future null infinity, the
quantum states that formed the black hole would have to be extremely energetic -- in fact transplanckian -- due to the gravitational redshift. Moreover, the stimulated emission effect is transient, and soon after stimulation the outgoing radiation becomes thermal again. Late-time quantum states do not suffer from this dramatic red shift, and the stimulated emission from late-time particles incident on a black hole should be readily observable at future null infinity}.

The realization that the black hole channel is the Unruh channel (even though only in the limit of a perfectly reflecting black hole) opens the door to study quantitatively the quantum information transmission properties of a black hole. Since the Unruh isometry gives rise to the Unruh channel whose quantum information properties are well understood~\cite{bradler2010conjugate,bradler2009private}, we can immediately transfer the properties of the Unruh channel to the black hole scenario in the perfectly reflecting black hole case. Using a dual-rail encoding of quantum information as in Eq.~(\ref{eq:UnruhDoubleIso}) we define
\begin{equation}\label{eq:HawkDoubleIso}
V_\om^{A_{1}\rightarrow B_{1}E_{1}}\otimes V_\om^{A_{2}\rightarrow B_{2}E_{2}}=V_\om^{A_{12}\rightarrow B_{12}E_{12}},
\end{equation}
where $V_\om^{A_{1}\rightarrow B_{1}E_{1}}$ is the isometry Eq.~(\ref{eq:HawkIso}). By tracing over the reference Hilbert space $E_{12}$, we obtain the {\em  black hole channel} whose form is identical to the Unruh channel Eq.~(\ref{eq:UnruhOutput}):
\begin{equation}\label{eq:BHchannel}
\N=\bigoplus_{\ell=1}^\infty p_\ell\ \Cl{\ell},
\end{equation}
so that the properties of the black hole channel are identical to the qubit Unruh channel, including the classical and quantum capacity~\cite{bradler2010conjugate,bradler2011infinite,bradler2012rindler}. The quantum capacity of the black hole channel $\N$ is therefore identical to Eq.~(\ref{eq:UnruhQCap}) and is depicted in Fig.~\ref{fig:BHcapacity}.
\begin{figure}[t]
    \begin{center}
        \resizebox{11cm}{6.8cm}{\includegraphics{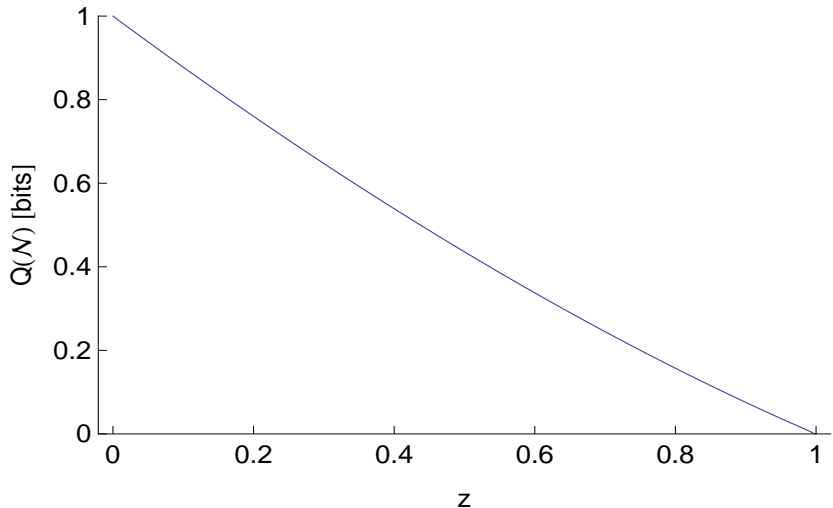}}
    \caption{The quantum capacity of the black hole channel Eq.~(\ref{eq:BHchannel}) where $z=\tanh^2{r_\om}$.}
    \label{fig:BHcapacity}
    \end{center}
\end{figure}
We note here that because we omitted the gravitational redshift in this derivation, the black hole channel for the case of perfect reflection of late-time quantum states is formally identical to the original Hawking channel given by the transformation
\be
a_{\rm out}=\alpha a-\beta b^\dagger\;, \label{hawking}
\ee
describing early-time modes. That the latter channel describes perfect reflection is obvious as particles sent towards the black hole in mode $a$ always remain in mode $a$, as these modes are defined precisely as those that travel towards future null infinity just outside the event horizon.  While Hawking was able to introduce grey-body factors using a transformation of the type (\ref{hawking}), we remind the reader that Bekenstein and Meisels~{\cite{bekenstein1977einstein} have shown that this form is not consistent with Einstein's formulation of the quantum theory of radiation~\cite{Einstein1917}. It was shown later by Adami and ver Steeg~\cite{AdamiVersteeg2013} that Sorkin's formulation with an arbitrary $\alpha$ gives rise to Hawking's result (including grey-body factors)
 \be
 \alpha^2=\frac{\Gamma}{1-e^{-\omega/T}}
 \ee
where $\Gamma$ is the black hole absorptivity and $T=\frac{\kappa}{2\pi}$ is the Hawking temperature, but in a thermodynamically consistent manner where $\Gamma$ is strictly smaller than 1, in perfect accord with~\cite{bekenstein1977einstein}. If the gravitational red shift is taken into account, the early-time and late-time isometries will be different, of course.

\subsection*{Entanglement-breaking channels}
One of the remarkable properties of the Unruh channel -- and by analogy therefore also of the perfectly reflecting black hole channel -- is the fact that its complementary channel is {\em entanglement-breaking}~\cite{bradler2011infinite}.
The complementary channel $\wh\N$ is the channel that relays information to the environment rather than the receiver. Its explicit form is known as well~\cite{bradler2011infinite}:
$$
\wh{\N}=\bigoplus_{\ell=1}^\infty p_\ell\whCl{\ell},
$$
where the $\whCl{\ell}$ denotes the complementary channel to the cloning channel $\Cl{\ell}$. Explicitly, we have
\begin{equation}\label{eq:BHchannelComplement}
  \whCl{\ell}(\phi)={2\over\ell(\ell+1)}\bigg({\ell+1\over2}\id_{\ell}+\sum_{i=x,y,z}m_iJ^{(\ell)}_i\bigg)
\end{equation}
and $m_x=n_x,m_y=-n_y$ and $m_z=n_z$ where $n_i$ has been defined in Eq.~(\ref{eq:cloningoutput}). For the black hole channel, the complementary channel connects Alice to the inside of the black hole. A quantum channel $\M$ is entanglement-breaking if $\vr_{AB}=(\id_A\otimes\M_B)(\vp_{AB})$ is separable for all entangled  bipartite states $\vp_{AB}$~\cite{horodecki2009quantum}. It is possible to prove rigorously that if a channel is entanglement-breaking, then its quantum capacity must be zero~\cite{hayden2008decoupling,klesse2008random}.
Thus, no quantum information can reliably be sent through an entanglement-breaking channel\footnote{The opposite implication does {\em not} hold: if a channel has zero capacity, it is not necessarily entanglement breaking.}.

The result we just derived, namely that no quantum information can enter a perfectly reflecting black hole, is satisfactory since at least for the external observer the only place that quantum information can go is in front of the horizon.  Furthermore, if the capacity to reconstruct quantum states perfectly outside the black hole is non-vanishing, we should expect that the quantum state inside of the horizon (which consists of anti-clones $\whCl{\ell}$ of the quantum state reflected on the outside~\cite{AdamiVersteeg2006}) cannot be used to reconstruct the quantum state, so that the no-cloning theorem is inviolate. This is precisely what we find.

What is the interpretation of the black hole channel capacity depicted in Fig.~\ref{fig:BHcapacity}? First of all, we see that the quantum capacity is nonzero for all values of the surface gravity  $\kappa=-2\pi\om/\log_2{z}$ except when $\kappa\to\infty$. The latter corresponds to the final stages of evaporation (microscopic black holes) which we will ignore due to the near certain breakdown of the semiclassical description. As long as the capacity stays nonzero (however small), the near-perfect transmission of quantum information is possible in the sense discussed earlier. One just has to wait longer since a lower quantum capacity implies a lower rate of quantum information transmission.  As the surface gravity increases so does the temperature, and the thermal Hawking background renders the transmission more noisy. Thus as long as the transmission rate decreases but stays nonzero, perfect quantum transmission is possible during all stages of the evaporation process until the semiclassical transformation~(\ref{eq:SorkinHamiltonian}) ceases to be a good description of the interaction. We also note that if stimulated emission is neglected (as in the original formulation of Hawking), then quantum information cannot be reconstructed even for the perfectly reflecting black hole as we have argued below Eq.~(\ref{eq:HawkDensityMatrix}).

Stimulated emission is ubiquitous in nature and appears in many elementary quantum systems. Mandel~\cite{mandel1983photon} presented a simple Hamiltonian describing an interaction of a two-level atom and a photon in an unknown polarization state. In today's language the photon is a polarization qubit and the Hamiltonian induces a completely positive map that turns out to be the optimal $1\to2$ cloner $\Cl{2}$.  The presence of a cloning transformation  suggests that a black hole is instead a rather ordinary quantum object (at least when it is still macroscopic).

Still, a reader might complain that a perfectly reflecting black hole does not correspond to a physical black hole. In the next section we will treat therefore the opposite extreme: a perfectly absorbing black hole from which no quantum states can be reflected. We focus on these extreme cases because only for those are single-letter quantum capacities calculable at the moment.

\subsection*{Perfectly absorbing black hole}

For a perfectly absorbing black hole, the Bogoliubov transformation (\ref{eq:SorkinHamiltonian}) can be implemented
by the Hamiltonian
\begin{equation}\label{sorkinham}
H_S = ig_\omega(a^\dg b^\dg-ab+a^\dg c-a c^\dg)
\end{equation}
so that
\begin{equation}
 a_{out}=e^{-iH_S}ae^{iH_S}=a- \gom(b^\dg+c).
\end{equation}
This corresponds to $\a=1$ in Eq.~(\ref{eq:SorkinHamiltonian}) which implies $\g^2=\b^2\equiv \gom^2$. Note that the
interaction between $c$-modes and $a$-modes in Eq.~(\ref{sorkinham}) takes the form of an ordinary beam splitter in quantum optics~\cite{Leonhardt2003}, and provides a way to describe the interaction of the black hole with radiation~\cite{AdamiVersteeg2006}.

In order to discover the nature of the  quantum channel that corresponds to the full absorption case, we will follow the same strategy as for the Unruh channel that described the perfectly reflecting black hole. Let us define an isometry
\begin{equation}\label{eq:fullAbsIso}
  V_\om\df\exp{(-i H_S)}.
\end{equation}
The input quantum information will be dual-rail encoded in the late-time mode $c$ and the output of the black hole channel will be collected by an outside observer  in  mode $a_{out}$. Hence we need an explicit action of $V_\omega$ for $n=0,1$. A tedious but straightforward calculation~\cite{AdamiVersteeg2013} leads to
\begin{equation}\label{eq:fullAbsIsoOnZero}
  V_\om\ket{000}_{abc}=\frac2{2+\gom^2}\sum_{n=0}^\infty\sum_{k=0}^{n}A^{n-k}B^k\sqrt{\binom{n}{k}}\ket{n-k}_a \ket{n}_b\ket{k}_c
\end{equation}
and
\begin{align}\label{eq:fullAbsIsoOnOne}
   V_\om\ket{001}_{abc} & ={\left(\frac2{2+\gom^2}\right)^2}\sum_{n=0}^\infty\sum_{k=0}^{n+1}A^{n-k}B^k\sqrt{\binom{n}{k}} \nonumber\\
  &\times \big(\sqrt{k+1}\ket{n-k}_a\ket{n}_b\ket{k+1}_c+ \gom\sqrt{n-k+1}\ket{n-k+1}_a\ket{n}_b\ket{k}_c\big),
\end{align}
where
\be
A=\frac{2\gom}{2+\gom^2}
\ee
and
\be
B=-\frac{\gom^2}{2+\gom^2}.
\ee
As before, the subscripts $a,b,c$ refer to the early-time modes ($a$ and $b$) and late-time modes ($c$) respectively. We can now define the dual-rail isometry as in the perfectly reflecting case:
$$
V_\om^{A_1\to B_1E_1}\otimes V_\om^{A_2\to B_2E_2}=V_\om^{A_{12}\to B_{12}E_{12}},
$$
where  the channel input system $A$ denotes the $c$ mode and carries the late-time quantum message encoded as a dual-rail qubit $\phi_A$ as in Eq.~(\ref{eq:dualrailInput}), the channel output by $B_{12}$ corresponds to the $a$ mode and the channel complementary output $E_{12}$ is the $b$ and $c$ mode. After tracing over the $B$ subsystem we finally obtain the corresponding black hole channel for full absorption:
\begin{equation}\label{eq:FullAbsChannel}
  \M=\bigoplus_{\ell=1}^\infty p_\ell\,\D_\ell
\end{equation}
for some probability distribution $p_\ell$. The channel has  a block-diagonal structure but it is a very different channel from the perfectly reflecting case Eq.~(\ref{eq:BHchannel}). The first output given by $\D_1$ (discussed below) reveals that this channel is a qubit depolarizing channel~\cite{king2003capacity}.

The family of depolarizing channels is a well studied topic in quantum Shannon theory~\cite{king2003capacity,nielsen2010quantum}}. A  qubit depolarizing channel can be expressed as
\begin{equation}\label{eq:DepChannel}
  \D_1(\vr)=(1-q)\vr+{q\over2}\id_2
\end{equation}
defined for $0\leq q\leq4/3$. The properties of the depolarizing channel crucially depend on $q$. For the perfectly absorbing  black hole we find $q=2/3$, which implies that the channel is entanglement-breaking. We can see this by noting that the partial transpose of the output density matrix of $\vr_{AB}=(\id_A\otimes\D_{1,B})(\Phi_{AB})$ is positive definite, where $\Phi_{AB}$ is a maximally entangled state. This is a necessary and sufficient condition for the output state to be separable~\cite{nielsen2010quantum} and so the channel is entanglement-breaking~\cite{horodecki2009quantum}.
This has profound consequences for the entire black hole channel $\M$ -- it must also be entanglement-breaking~\cite{bradler2011infinite}. Indeed, we first observe that
\begin{equation}
  \D_1(\phi)={1\over2}\id_2+\sum_{i=x,y,z}k_iJ_i^{(2)},
\end{equation}
where
\begin{subequations}
\begin{align}
  k_x &= (q-1)(a\bar b+\bar ab), \\
  k_y &= (q-1)i(a\bar b-\bar ab), \\
  k_z &= (q-1)(|a|^2-|b|^2),
\end{align}
\end{subequations}
and $J_i^{(2)}$ are the  generators of the fundamental representation of the $su(2)$ algebra. But similarly to the Unruh channel structure, all the blocks in Eq.~(\ref{eq:FullAbsChannel}) can be written as
\begin{equation}\label{eq:depoloutput}
  \D_{\ell}(\phi)={2\over\ell(\ell+1)}\bigg({\ell\over2}\id_{\ell+1}+\sum_{i=x,y,z}k_iJ^{(\ell+1)}_i\bigg).
\end{equation}
Following the argument in~\cite{bradler2011infinite} we conclude that all $\D_\ell$ are entanglement-breaking and so is the channel $\M$.

We thus arrive at the conclusion that the black hole channel (for the case of perfect absorption) is entanglement-breaking, and therefore  the quantum capacity vanishes. This observation has a number of interesting consequences. First of all, being unable to transmit entanglement across a quantum depolarizing channel whose fidelity is below a critical level is not a violation of any law of physics. The real question concerns the fate of this quantum information after full evaporation of the black hole.

\section{Discussion}
\label{sec:concl}

Black holes are quantum objects: Hawking's derivation of the radiation effect has surely taught us this much. But what kind of objects are they when considered in light of quantum information theory, and in particular quantum Shannon theory? Here we argue that black holes act as a depolarizing medium that, depending on the reflectivity of the black hole, can obstruct the perfect reconstruction of quantum states sent through the channel.
We have derived the quantum capacity for the transmission of entanglement via black holes in two important limiting cases: perfectly reflecting black holes (which can be seen as white holes~\cite{BirrellDavies1982}), and perfectly absorbing black holes. Radiation impinging on a perfectly reflecting black hole creates two  clones of the quantum information on the outside and a single anti-clone\footnote{We remind the reader that the terms ``clone" and ``anti-clone" refer to approximations of the original quantum states (as described on p.~\pageref{cloningdisc}).} traveling towards the black hole singularity. We show that Bob can reconstruct Alice's quantum state (that is, obtain the same entanglement with respect to a reference state that Alice had) with perfect accuracy by suitably acting on the clones (following from the non-vanishing capacity of the Unruh channel), while it is not possible to reconstruct the quantum state perfectly from the anti-clone behind the horizon, because that channel turns out to be entanglement-breaking (with zero quantum capacity). As a by-product, the no-cloning theorem is inviolate. We note that, viewed from the {\em inside} of the black hole, the channel is perfectly reflecting: no quantum information can leak outside.

The black hole channel for perfect absorption appears to be complementary to the perfect reflection channel. We show that the capacity to reconstruct quantum information outside of the black hole vanishes: the channel is entanglement breaking. At the same time, the absorbed quantum states can be used to perfectly reconstruct entanglement behind the horizon, as now there are {\em two clones} (actually a clone and an anti-clone) of the quantum state behind the horizon, but only a single (approximate) clone in front of the horizon.
\section{Conclusions}

The consequences of this analysis of black holes in terms of quantum channel theory are manifold. We separated the notions of quantum information loss from the breakdown of quantum mechanics during the black hole evolution. This has been achieved by pointing at the precise meaning of quantum information as established in quantum information (Shannon)  theory. This field also reminded us that the unitary evolution of an open system is not the most general dynamics allowed by the laws of quantum mechanics. We argue that  black holes are open quantum systems and as such are allowed to evolve from pure states into mixed states. This helps us in describing the black hole as a quantum channel (completely positive map) and to finally calculate its quantum capacity. The latter quantifies the amount of quantum information that can be transmitted through a black hole as understood in current quantum information theory. We can go even further: since quantum channels with zero quantum channel capacity are certainly allowed, we could end up in a situation where quantum information is simply lost. Whether this happens depends on the ultimate fate of black holes (which is not known at present). Even if quantum information turns out to be ultimately unrecoverable (in the perfectly absorbing scenario), such
a loss of quantum information does not violate any known law. We note, however, that the present quantum information theoretic treatment guarantees that the no-cloning theorem of quantum mechanics is respected, something that cannot be guaranteed in scenarios such as black hole complementarity~\cite{susskind1993stretched}.

Perhaps a closer look at black hole complementarity is warranted, then. Susskind et al.~\cite{susskind1993stretched} propose to solve the information paradox by positing that quantum information can both be reflected from {\em and} transmitted through the event horizon, and that the concomitant violation of the no-cloning theorem is averted simply because it is not possible to ascertain such a violation experimentally. What we show here is that (at least for the case of perfect absorption and perfect reflection), quantum information is not both reflected and transmitted. While in the perfectly reflecting channel an anti-clone is indeed transmitted into the black hole (while the quantum information is reflected), the anti-clone is insufficient to resurrect the quantum state, so that the no-cloning theorem is not violated.

On the other hand, if the channel is perfectly absorbing then a clone is indeed ``reflected" (stimulated), but again is insufficient to reconstruct the quantum state, while the quantum state and its anti-clone have disappeared behind the horizon, carrying with them Alice's quantum entanglement. Thus, there is a complementarity within quantum black hole channels, but it is perfectly in accord with our laws of physics, unlike the quantum-information-theoretically naive interpretation of Ref.~\cite{susskind1993stretched}. While we have only demonstrated this complementarity here for the two extreme cases of the black channel, we expect that it holds in the most general case.

\section*{Acknowledgements}
KB enjoyed discussions with Luigi Gallo about observational aspects of black holes.

\bibliographystyle{unsrt}


\section*{APPENDIX}
\label{sec:appendix}

\renewcommand{\theequation}{A.\arabic{equation}}
\setcounter{equation}{0}

Here we calculate the quantum capacity of a direct sum channel given the quantum capacity of the summands.

\begin{lem}
\label{lem:directsum} Consider quantum channels $\N_{1},\N_{2}$ and $\T$ and suppose $Q^{(1)}(\N_{i}\otimes\T)\leq Q^{(1)}(\mathcal{N}_{i})+Q^{(1)}(\mathcal{T})$,
where $Q^{(1)}$ is the optimized coherent information. Then
\begin{equation}\label{eq:inass}
Q^{(1)}((p_{1}\N_{1}\oplus p_{2}\N_{2})\otimes\T)
\leq Q^{(1)}(p_{1}\N_{1}\oplus p_{2}\N_{2})+Q^{(1)}(\T),
\end{equation}
where $p_{1}+p_{2}=1$ and $\N=p_{1}\N_{1}\oplus p_{2}\N_{2}$ is a direct sum channel whose output is a
classical-quantum state.
\end{lem}

\begin{proof}
We denote $\vr_{BC}=(\N\otimes\T)(\phi)=\sum_{x=1,2}\kbr{x}{x}\otimes\vr_{B_xC}$, where $B=B_1\oplus B_2$. First note that the complementary channel of  a direct sum channel $\N=p_1\N_1\oplus p_2\N_2$ can be written as a direct sum of the complementary channel  $\wh\N_1$ and $\wh\N_2$ with the same probabilities (the caret denotes the complementary channel or subsystem). This immediately follows from the purification of the classical-quantum state $\vr_{BC}$. We can then write the complementary channel output as another classical-quantum state $\s_{\wh B\wh C}=(\wh\N\otimes\wh\T)(\phi)=\sum_{x=1,2}\kbr{x}{x}\otimes\s_{\wh B_x\wh C}$, where $\wh B=\wh B_1\oplus\wh B_2$. It follows that:
\begin{subequations}
\begin{align}
Q^{(1)}(\N\otimes\T) & = \max_\phi{\big[I_c(\N\otimes\T,\vr_{BC})\big]} \\
& = \max_\phi{\big[I_c((p_1\N_1\oplus p_2\N_2)\otimes\T,\vr_{BC})\big]} \\
& = \max_\phi{\big[H(BC)_\vr-H(\wh B\wh C)_\s\big]} \\
& = \max_\phi{\big[H(\{p_1,p_2\})+p_1H(B_1C)_\vr+p_2H(B_2C)_\vr} \label{eq:directsumQ14}  \\
& \ \ \ \ \ \ -(H(\{p_1,p_2\})+p_1H(\wh B_1\wh C)_\s+p_2H(\wh B_2\wh C)_\s)  \big] \\
& = \max_\phi{\big[p_1(H(B_1C)_\vr-H(\wh B_1\wh C)_\s)+p_2(H(B_2C)_\vr}-H(\wh B_2\wh C)_\s)  \big] \\
& = p_1Q^{(1)}(\N_1\otimes\T)+p_2Q^{(1)}(\N_2\otimes\T)\\
& \leq p_1Q^{(1)}(\N_1)+p_2Q^{(1)}(\N_2)+Q^{(1)}(\T) \label{eq:directsumQ17} \\
& = \max_{\phi'}{\big[H(\{p_1,p_2\})+p_1H(B_1)_{\vr'}+p_2H(B_2)_{\vr'}} \\
& \ \ \ \ \ \ -(H(\{p_1,p_2\})+p_1H(\wh B_1)_{\s'}+p_2H(\wh B_2)_{\s'})  \big] + Q^{(1)}(\T) \\
& = \max_{\phi'}{\big[H(B)_{\vr'}-H(\wh B)_{\s'}\big]} + Q^{(1)}(\T) \\
& = Q^{(1)}(p_1\N_1\oplus p_2\N_2)+Q^{(1)}(\T)\nn.
\end{align}
\end{subequations}
The first three lines follow directly from the definition of the optimized coherent information. The first crucial equality is Eq.~(\ref{eq:directsumQ14}) where we used the direct sum structure of the output state occupying the $B$ subsystem. In~Eq.~(\ref{eq:directsumQ17}) we used the initial assumption Eq.~(\ref{eq:inass}) and the rest is again just the definition of the optimized coherent information. Note that the states $\phi'$ and $\vr',\s'$ are convenient states not necessarily related to $\phi,\vr,\s$.
\end{proof}
In particular, if $\T=\N=p_{1}\N_{1}\oplus p_{2}\N_{2}$ we have shown that the quantum capacity of
the quantum channel $\N$ reads
\[
Q(p_{1}\mathcal{N}_{1}\oplus p_{2}\mathcal{N}_{2})=Q^{(1)}(p_{1}
\mathcal{N}_{1}\oplus p_{2}\mathcal{N}_{2})=p_{1}Q^{(1)}(\mathcal{N}
_{1})+p_{2}Q^{(1)}(\mathcal{N}_{2})=p_{1}Q(\mathcal{N}_{1})+p_{2}
Q(\mathcal{N}_{2}).
\]
A simple inductive argument yields the following expression
\begin{equation}
Q\bigg(\bigoplus_{i=1}^{n}p_{i}\mathcal{N}_{i}\bigg)=\sum_{i=1}^{n}
p_{i}Q(\mathcal{N}_{i}), \label{eq:QdirectSumChannel}
\end{equation}
valid for $n>0$ assuming $\sum_{i=1}^{n} p_{i}=1$ (including $n\to\infty$).

Lemma~\ref{lem:directsum} shows how to calculate the quantum capacity of a direct sum channel such as the black hole channel in $\N$~Eq.~(\ref{eq:BHchannel}).

\end{document}